\documentclass{tran-l}
\usepackage{eurosym}
\usepackage{amsmath, amssymb}
\usepackage{graphicx}
\usepackage{color}
\usepackage{textcomp}

\usepackage{graphics}

 \newtheorem{thm}{Theorem}[section]

 \newtheorem{prop}[thm]{Proposition}
 \theoremstyle{definition}
 
 \theoremstyle{remark}
 \newtheorem{Remark}[thm]{Remark}
 
 \newcommand{\PP}{\mathbb{P}}

\newcommand{\bm}{\bibitem}
\newcommand{\no}{\noindent}
\newcommand{\be}{\begin{equation}}
\newcommand{\ee}{\end{equation}}

\newcommand{\bea}{\begin{eqnarray}}
\newcommand{\bes}{\begin{subequations}}
\newcommand{\ees}{\end{subequations}}
\newcommand{\bgt}{\begin{gather}}
\newcommand{\egt}{\begin{gather}}
\newcommand{\eea}{\end{eqnarray}}
\newcommand{\beaa}{\begin{eqnarray*}}
\newcommand{\eeaa}{\end{eqnarray*}}

\newcommand{\NN}{{\mathbb N}}
\newcommand{\EE}{{\mathbb E}}
\newcommand{\RR}{{\mathbb R}}
\newcommand{\cal}{\mathcal}

\begin{document}

\title{Counterparty Risk Valuation: \\ A Marked Branching Diffusion Approach}
\author{Pierre Henry-Labord\`ere}
\address{Soci\'et\'e G\'en\'erale, Global Market Quantitative
Research} \email{pierre.henry-labordere@sgcib.com}
\date{}

\maketitle

\begin{abstract}
\no The purpose of this paper is to design an algorithm for the
computation of the counterparty risk which is competitive in regards
of a brute force ``Monte-Carlo of Monte-Carlo"  method (with nested
simulations). This is achieved using marked branching diffusions
describing a Galton-Watson random tree. Such an algorithm leads at
the same time to a computation of the (bilateral) counterparty risk
when we use the default-risky or counterparty-riskless option values
as mark-to-market. Our method is illustrated by various numerical
examples.
\end{abstract}

\section{Introduction}

\no The recent financial crisis has highlighted the importance of
credit valuation adjustment when pricing derivative contracts.
Bilateral counterparty risk is the risk that the issuer of a
derivative contract or the counterparty may default prior to the
expiry and fail to make future payments. This market imperfection
leads naturally for Markovian models to non-linear second-order
parabolic partial differential equations (PDEs). More precisely, the
non-linearity in the pricing equation affects none of the
differential terms and depends on the positive part of the
mark-to-market value of the derivative upon default. We have a
so-called semi-linear PDE. The numerical solution of this equation
is a formidable task that has attracted little attention from
practitioners. For multi-asset portfolios, these PDEs which suffer
from the curse of dimensionality cannot be solved with
finite-difference schemes. We must rely on probabilistic methods. Up
to now, it seems that a brute force intensive ``Monte-Carlo of
Monte-Carlo"  method (with nested simulations) is the only tool
available for this task.

\no In this paper, we rely on new advanced non-linear Monte-Carlo
methods for solving these  semi-linear PDEs. A first approach is to
use the so-called first-order backward stochastic differential
equations. Unfortunately, in practise this method requires the
computation of conditional expectations using regressions. Finding
good quality regressors is notably difficult, especially for
multi-asset portfolios. This leads us to introduce a new method
based on branching diffusions describing a marked Galton-Watson
random tree. A similar algorithm can also be applied to obtain
stochastic representations for solutions of a large class of
semi-linear parabolic PDEs in which the non-linearity  can be
approximated by a polynomial function.

\section{Credit Valuation Adjustment}

\subsection{Semi-linear PDEs}

\no For completeness, we derive the PDE arising in counterparty risk
valuation of a European derivative with a payoff $\psi$ at maturity
$T$. In short, depending on the (modeling) choice of the
mark-to-market value of the derivative upon default, we will get two
types of semi-linear PDEs that can be schematically written as \bea
\partial_t u + {\cal L} u +r_0 u +r_1 u^+ &=& 0,\quad u(T,x)=\psi(x)
\label{noMkt}\eea \no and \bea
\partial_t u + {\cal L} u +r_0 u + r_1 M +r_2 M^+  &=& 0,\quad \label{Mkt}
u(T,x)=\psi(x)  \\
\partial_t M + {\cal L} M +r_4 M  &=& 0,\quad M(T,x)=\psi(x) \nonumber
\eea ${\cal L}$ is the It\^o generator of a multi-dimensional
diffusion process and $r_i$ are arbitrary functions of $t$ and $x$.

\subsection{PDE derivation}

\no We assume the issuer is allowed to dynamically trade $d$
underlying assets $X_\cdot \in \RR^d_+$. Additionally, in order to
hedge his credit risk on the counterparty name, he can trade a
default risky bond, denoted $P_t^{2}$. Furthermore, the values of
the underlyings are  not altered by the counterparty default which
is modeled by a Poisson jump process. For the sake of simplicity, we
consider a constant intensity. This assumption can be easily
relaxed, in particular the intensity can follow an It\^o diffusion.
For use below, expressions with a subscript $2$  denote counterparty
quantities. We consider the case of a long position in a single
derivative whose value we denote $u$. In practice netting agreements
apply to the global mark-to-market value of a pool of derivative
positions - $u$ would then denote the aggregate value of these
derivatives. The processes $X_t, P_t^{2}$ satisfy under the
risk-neutral
measure $\PP$ (we assume the market model is complete) \beaa {dX_t \over X_t}&=&r dt+\sigma(t,X_t).dW_t \\
{dP_t^{2} \over P_t^{2}}&=&(r+\lambda_2)dt-dJ_t^2 \eeaa \no with
$W_t$ a $d$-dimensional Brownian motion, $J_t^2$ a jump Poisson
process with intensity $\lambda_2$ and $r$ the interest rate. The
no-arbitrage condition and the completeness of the market give that
$e^{-rt}u(t,X_t)$ is a $\PP$-martingale, characterized by
 \beaa &&\partial_t u +  {\cal L} u  + \lambda_2 \left(
\tilde{u}- u\right)
 -r u    =0 \eeaa where ${\cal L}$ denotes the It\^o
generator of $X$ and $\tilde{u}$
 the derivative value after the counterparty
has defaulted. At the default event, $\tilde{u}$ is given
by\footnote{$X \equiv X^+- X^-$.} \beaa \tilde{u}&=&R M^+-M^- \eeaa
with $M$ the mark-to-market value of the derivative  to be used in
the unwinding of the position upon default and $R$ the recovery
rate. There is an ambiguity in the market about the convention for
the mark-to-market value to be settled at default. There are two
natural conventions (see \cite{bri} for discussions about the
relevance of these conventions): The mark-to-market of the
derivative is evaluated at the time of default with provision for
counterparty risk or without.

\no $1$. Provision for counterparty risk, $M=u$: \bea &&\partial_t u
+ {\cal L} u - (1-R) \lambda_2 u^+  -ru=0,\quad u(T,x)=\psi(x) \eea
\no In the particular case when the payoff $\psi(x)$ is negative,
the solution is given by $e^{-r(T-t)}\EE_{t,x}[\psi(X_T)]$.

\no $2$. No provision for counterparty risk: \bea &&\partial_t u +
{\cal L} u
 +\lambda_2 \left( R M^+-M^- - u\right)
 -r u    =0,\quad u(T,x)=\psi(x)  \label{PDEsystem} \\
&&\partial_t M +  {\cal L} M  -r M   =0,\quad M(T,x)=\psi(x)
\nonumber \eea \no In the case of collateralized positions,
counterparty risk applies to the variation of the mark-to-market
value of the corresponding positions experienced over the time it
takes to qualify a failure to pay margin as a default event -
typically a few days. In the latter case, the non-linearity $u_t^+$
should be substituted with $(u_t-u_{t+\Delta})^+$ where $\Delta$ is
this delay. We will come back to this situation in the last section
(see remark \ref{colla}).

\no  By proper discounting and replacing $u$ by $-u$ for the sake of
the presentation, these two PDEs can be cast into normal forms \bea
&&
\partial_t u + {\cal L} u +\beta \left( u^+ -u\right)=0,\quad
u(T,x)=\psi(x)
\label{toy}:{\bf \mathrm{PDE2}}  \\
&&\partial_t u + {\cal L} u +{\beta \over 1-R} \left( (1-R)
\EE_{t,x}[\psi]^++R \EE_{t,x}[\psi] - u\right)=0,\quad
u(T,x)=\psi(x) \label{toy2}:{\bf  \mathrm{PDE1}}  \eea with $\beta
\equiv \lambda_2(1-R) \in \RR^+$. \no It is interesting to note that
a similar semi-linear PDE type (\ref{toy}) appears also in the
pricing of American options.

\subsection{American options}

\no The replication price of an American option with exercise payoff
$\psi(x)$ satisfies a variational PDE: \beaa \max\left( \partial_t u
+{\cal L} u , \psi(x)-u \right)=0,\quad u(T,x)=\psi(x) \eeaa \no
This PDE can be converted into a semi-linear PDE (see \cite{ben} for
details): \beaa
\partial_t u +{\cal L} u =1_{\psi(x) \geq u}{\cal L} \psi(x),\quad
u(T,x)=\psi(x) \eeaa Stochastic representations of this equation
lead to well-known early exercise premium formulas of American
options. Our algorithm can also be applied to this non-linear PDE.
It does not require regressions as in the well-known
Longstaff-Schwartz method \cite{lon} or a ``Monte-Carlo of
Monte-Carlo method" as in Rogers's dual algorithm \cite{bro, rog}.
\vskip 2truemm \no In the next section, we briefly list (non-linear)
Monte-Carlo algorithms which can be used to solve PDEs
(\ref{toy})-(\ref{toy2}) and highlight their weaknesses in the
context of credit valuation adjustment.

\section{Non-linear Monte-Carlo algorithms}

\subsection{A brute force algorithm}

\no Using Feynman-Kac's formula, the solution of PDE (\ref{toy}) can
be represented stochastically as \bea
 u(t,x)= e^{-\beta(T-t)} \EE_{t,x}[\psi(X_T)]  +\int_t^T \beta e^{-\beta(s-t)}\EE_{t,x}[u^+(s,X_s)]ds
 \eea with $X$ an It\^o diffusion with generator ${\cal L}$ and $\EE_{t,x}[\cdot]=
 \EE[\cdot|X_t=x]$. By
 assuming that the intensity $\beta$ is small, we get the
 approximation (this is exact for PDE (\ref{toy2})\footnote{Precisely, we get $
e^{-{\lambda_2}(T-t)} \EE_{t,x}[\psi(X_T)]  +\lambda_2
 \int_t^T e^{-\lambda_2(s-t)}
 \EE_{t,x}[(1-R)\left(\EE_{s,X_s}[\psi(X_T)]\right)^++R\EE_{s,X_s}[\psi(X_T)]]ds$.} )
\bea
 u(t,x) = e^{-\beta(T-t)} \EE_{t,x}[\psi(X_T)]  +\beta e^{-\beta(T-t)} \int_t^T
 \EE_{t,x}[\left(\EE_{s,X_s}[\psi(X_T)]\right)^+]ds +O(\beta^2)
 \label{approxbeta} \eea Then, at a next step, we discretise the Riemann integral
\beaa
 u(t,x) \simeq e^{-\beta(T-t)}  \EE_{t,x}[\psi(X_T)]  + \beta e^{-\beta(T-t)}  \sum_{i=1}^n
 \EE_{t,x}[\left(\EE_{t_i,X_{t_i}}[\psi(X_T)]\right)^+] \Delta t_i
 \eeaa \no This last expression can be numerically tackled by using
 a brute force ``Monte-Carlo of Monte-Carlo" method. The second MC is used
 to compute $\EE_{t_i,X_{t_i}}\psi(X_T)]$ on each path generated by the
 first MC algorithm. Although straightforward, this method
 suffers from the curse of dimensionality and requires generating
 $O(N_1 \times N_2)$ paths. Due to this complexity, the literature focuses on
 exposition of linear portfolios for which the second MC can be
 skipped by using  closed-form formulas or low-dimensional parametric
 regressions (see for example \cite{bri0} in which the authors consider the pricing of CMS spread option and CCDSs).

 \no Could we design a simple (non-linear) Monte-Carlo algorithm
which solves our PDEs (\ref{toy})-(\ref{toy2}), without relying on
an approximation such as (\ref{approxbeta})? This is the purpose of
this paper.

\subsection{Backward stochastic differential equations}

\no A first approach is to simulate a backward stochastic
differential equation (in short BSDE): \bea
dX_t&=&\mu(t,X_t)dt+\sigma(t,X_t).dW_t,\quad X_0=x \\
dY_t&=&-\beta Y_t^+dt+Z_t \sigma(t,X_t).dW_t \\
Y_T&=&\psi(X_T) \label{terminal} \eea \no where $(Y,Z)$ are required
to be adapted processes and ${\cal L}=\sum_i \mu_i\partial_{x^i}
 + {1 \over 2} \sum_{i,j} (\sigma \sigma^*)_{ij}\partial^2_{x^i x^j}$. BSDEs differ from
(forward) SDEs in that we impose the terminal value (see Equation
(\ref{terminal})). Under the condition $\psi \in
\mathrm{L}^2(\Omega)$, this BSDE admits a unique solution
\cite{par}. A straightforward application of It\^o's lemma gives
that the solution of this BSDE is $\left(Y_t=e^{\beta(T-t)}u(t,X_t),
Z_t=e^{\beta(T-t)} \nabla_x u(t,X_t)\right)$ with $u$ the solution
of PDE (\ref{toy}). This leads to a Monte-Carlo like numerical
solution of (\ref{toy}) via an efficient discretization scheme for
the above BSDE.

\no This BSDE can be discretized by an Euler-like scheme
($Y_{t_{i-1}}$ is forced to be ${\cal F}_{t_{i-1}}$-adapted, $({\cal
F}_t)_{t \geq 0}$ being the natural filtration generated by the
Brownian motions): \beaa \EE_{t_{i-1}}[Y_{t_{i}}]-Y_{t_{i-1}}=-\beta
\Delta t_i \left( \theta Y_{t_{i-1}}^+ +(1-\theta)
\EE_{t_{i-1}}[Y_{t_{i}}]^+ \right)\eeaa \no with $\theta \in [0,1]$.
This is equivalent to (we take $\theta \beta \Delta t_i<1$) \beaa
Y_{t_{i-1}}&=&\EE_{t_{i-1}}[Y_{t_{i}}] \left(
1_{\EE_{t_{i-1}}^\PP[Y_{t_{i}}]>0}{1+(1-\theta)\beta \Delta t_i
\over 1- \theta \beta \Delta t_i} +1_{\EE_{t_{i-1}}[Y_{t_{i}}]<0}
\right) \eeaa \no This requires the computation of the conditional
expectation $\EE_{t_{i-1}}^\PP[Y_{t_{i}}]$ (in practise by
regression methods) which could be quite difficult and
time-consuming, especially for multi-asset portfolios.

\subsection{Gradient representation}

\no A more powerful approach is synthesized by the following
proposition which relies on Kunita's stochastic flows of
diffeomorphisms (see \cite{tal}). Let $u$ be the solution of the
\emph{one-dimensional} semi-linear PDE
 \bea \partial_t u+{1 \over 2}\sigma^2(t,x) \partial_x^2 u+
 f(u)=0 \label{1dnpde} \eea with the terminal condition $u(T,x)=\psi(x)$. By differentiating equation (\ref{1dnpde}) with respect to $x$
(assuming smoothness of the coefficients) we get \bea &&\partial_t
\Delta+\left( \left(\sigma \partial_x \sigma \right)\partial_x+ {1
\over 2}\sigma^2(t,x)
\partial_x^2\right)\Delta  +  f'\left(u\right)\Delta=0 \eea with the terminal condition
$\Delta(T,x)=\psi'(x)$. The equation satisfied by the gradient
$\Delta$ is then interpreted as a (linear) Fokker-Planck PDE. We
have the following representation \cite{tal}
  \beaa u(t,x)=-\int_{\RR^+} \psi'(a) da \; \EE_{t}[ 1(X_{T}^{a}-x) e^{\int_t^T f'\left(u(T+t-s,X_s^a)\right) ds}] \eeaa where
the It\^o process $X_s^a$ is the solution to \beaa
dX_s^a=\sigma(T+t-s,X_s^a)dB_s+\left(\sigma\partial_x
\sigma\right)(T+t-s,X_s^a)ds \;,\; s \in [t,T] \;,\; X_t^a=a \eeaa
\no $B_s$ is a standard Brownian. This representation leads to a
particle algorithm \cite{tal}. Although appealing, this (forward)
approach is only applicable in the one-dimension setup for which we
can use a PDE solver. Can we design a similar forward algorithm
applicable in higher dimensions? This leads us to
 branching diffusions.

\subsection{Branching diffusions: an introduction}
\label{branchingsection} \no Branching diffusions have been first
introduced by McKean \cite{MCKean} to give a probabilistic
representation of the Kolmogorov-Petrovskii-Piskunov PDE and more
generally of semi-linear PDEs of the type \bea &&\partial_t u +
{\cal L} u +\beta(t)\left(\sum_{k=0}^\infty
p_k u^k -u\right)=0\quad\mathrm{in}\quad \RR_+ \times \RR^d   \label{KPP} \\
&&u(T,x)=\psi(x)\quad\mathrm{in}\quad \RR^d  \nonumber \eea with
$\beta(\cdot) \in \RR^+$. Here the non-linearity is a power series
in $u$ where the coefficients satisfy the restrictive condition:
\bea f(u)\equiv \sum_{k=0}^\infty p_k u^k,\quad \sum_{k=0}^\infty
p_k=1\,\quad 0 \leq p_k \leq 1 \label{Poly} \eea \no The
probabilistic interpretation of such an equation goes as follows:
\no Let a single particle start at the origin, perform an It\^o
diffusion  on $\RR^d$ with generator ${\cal L}$, after a mean
$\beta(\cdot)$ exponential time (independent of $X$) die and produce
$k$ descendants with probability
 $p_k$ ($k=0$ means that the particle dies without generating descendants). Then, the descendants
perform independent It\^o diffusions on $\RR^d$ (with same generator
$\cal L$) from their birth locations, die and produce descendants
after a mean $\beta(\cdot)$ exponential times, etc. This process is
called a $d$-dimensional branching diffusion with a branching rate
$\beta(\cdot)$. $\beta$ can also depend spatially on $x$ or be
itself stochastic (Cox process). We note $Z_t \equiv \left(z_t^1 ,
\ldots, z_t^{N_t}\right) \in \RR^{d \times N_t}$ the locations of
the particles alive at time $t$ and $N_t$ the number of particles at
$t$ (see Fig. \ref{Fey} for examples with $2$ and $3$ descendants).
We consider then the multiplicative functional defined
by\footnote{$\prod^{N_T=0} \equiv 1$ by convention.} \bea
\hat{u}(t,x)=\EE_{t,x}\Big[ \prod_{i=1}^{N_T} \psi(z_T^i)\Big]
\label{mult} \eea where $\EE_{t,x}[\cdot]=\EE[\cdot|N_t=1,
z_t^1=x]$. Note that as $N_T$ can become infinite when
$m=\sum_{k=0}^\infty k p_k
>1$ (super-critical regime, see \cite{mel}), a sufficient condition on $\psi$ in
order to have a well-behaved product is $|\psi|<1$. Then $\hat{u}$
solves the semi-linear PDE (\ref{KPP}). This stochastic
representation can be understood as follows: Mathematically, by
conditioning on $\tau$, the first-time to jump of a Poisson process
with intensity $\beta(t)$, we get from (\ref{mult}) \beaa \hat{u
}(t,x)=\EE_{t,x}[1_{\tau  \geq T}\psi(z^1_T)] +\EE_{t,x}[ 1_{\tau<T}
\sum_{k=0}^\infty p_k \EE_\tau[ \prod_{j=1}^{k}
\prod_{i=1}^{N_T^j(\tau)} \psi(z_T^{i,j,z_\tau})] \eeaa where
$z_T^{i,j,z_\tau}$ is the position of the $i$-th particle at
maturity $T$ produced by the $j$-th particle generated at time
$\tau$. By using the independence and the strong Markov property, we
obtain \beaa \hat{u}(t,x)&=&\EE_{t,x}[1_{\tau \geq T}\psi(z^1_T)]
+\sum_{k=0}^\infty\EE_{t,x}[ 1_{\tau<T}  p_k \prod_{j=1}^{k}
\EE_\tau[ \prod_{i=1}^{N_T^j(\tau)}
\psi(z_T^{i,j,z_\tau})] \\
&=&\EE_{t,x}[1_{\tau \geq T}\psi(z^1_T)] +\EE_{t,x}[ 1_{\tau<T}
\sum_{k=0}^\infty
p_k \prod_{j=1}^{k} \hat{u}(\tau,z^1_\tau)] \\
&=&\EE_{t,x}[1_{\tau \geq T}\psi(z^1_T)] +\sum_{k=0}^\infty p_k
\EE_{t,x}[
 \hat{u}^k(\tau,z^1_\tau)1_{\tau<T}]  \\
 &=& \EE_{t,x}[ e^{-\int_t^T \beta(s)ds}
\psi(z_T^1)]+\int_t^T  \sum_{k=0}^\infty p_k \EE_{t,x}[\beta(s)
e^{-\int_t^s \beta(u)du} \hat{u}^k(s,z^1_s)]  ds  \eeaa Then, by
assuming that $||\psi||_\infty<1$, $\hat{u}$ is uniformly bounded by
$1$ in $[0,T] \times \RR^d$ and we get from the Feynman-Kac formula
that $\hat{u}$ is a viscosity solution to PDE (\ref{KPP}) (see
Theorem 6.4 in \cite{tou}). By assuming that PDE (\ref{KPP})
satisfies a comparison principle, we conclude that $u=\hat{u}$.

\no A first attempt in order to obtain a larger class of
non-linearities than those defined by (\ref{Poly}) is to consider an
infinite collection of branching diffusions, the so-called
super-diffusions. (\ref{Poly}) is then extended to \bea \Psi(u)=a u
+b u^2+\int_{0}^\infty n(dr)[ e^{-ru}-1+ru] \label{Poly1} \eea where
$a \geq 0$, $b \geq 0$ and $n$ is a Radon measure on $(0,\infty)$
satisfying $\int_0^\infty (r \wedge r^2) n(dr) <\infty$. The class
of non-linearity as defined by (\ref{Poly1}) is more general than
(\ref{Poly}), in particular contains $a u +b u^2$ with arbitrary
positive coefficients $a$ and $b$. Unfortunately, this requires a
large number of branching diffusions (as the default intensity
diverges) and the non-linearity is still restrictive. This leads us
to introduce a new class of branching diffusions that can be traced
back to Le Jan-Sznitman \cite{jan} in the context of stochastic
(Fourier) representations of solutions of the incompressible
Navier-Stokes equation.

\section{Marked branching diffusions}

\no The PDE (\ref{KPP}) should be compared with the semi-linear PDE
(\ref{toy}) arising in the pricing of counterparty risk. It seems
too restrictive and unreasonable to approximate the non-linearity
$u^+$ by a polynomial of type (\ref{Poly}) or even (\ref{Poly1}). A
natural question is therefore to search if this construction can be
generalized for an arbitrary polynomial for which the PDE is \bea
\partial_t u+{\cal L} u + \beta(F(u) -u)=0 \label{PDEpoly} \eea with
$F(u)=\sum_{k=0}^M a_k u^k$ an  $M$-order polynomial in $u$ that we
write for convenience $F(u)=\sum_{k=0}^M \left({a_k \over
p_k}\right) p_k u^k$. We will show below that this can be achieved
by counting the branching of each monomial $u^k$.

\no {\bf Assumption (Comp)}: In order to have  uniqueness in the
viscosity sense, we assume PDE (\ref{PDEpoly}) satisfies a
comparison principle for sub- and super-solutions (see \cite{fle}).

\no For each Galton-Watson tree, we denote $\omega_k \in \NN$ the
number of branching of monomial type $u^k$ with $k \in \{0, \ldots,M
\}$. The descendants are drawn with an arbitrary distribution $p_k$
- for example we can take  a uniform distribution $p_k={1 \over
M+1}$ (see an other choice in section \ref{sectionoptimal}). In Fig.
\ref{Fey}, we have drawn the diagrams for the non-linearity $F(u)={a
\over p_2} p_2 u^2 + {b \over p_3} p_3 u^3$ up to two defaults. \no
We then define the multiplicative functional:

\no{\bf Main formula}: \bea \hat{u}(t,x)=\EE_{t,x}\Big[
\prod_{i=1}^{N_T} \psi(z_T^i) \prod_{k=0}^M \left({a_k \over
p_k}\right)^{\omega_k} \Big] \;,\; \omega_k=\sharp\mathrm{branching
\; type} \; k\label{Rep1} \eea \no We state our main result (the
proof  is reported in the appendix):
\begin{thm} Let us assume that $\hat{u} \in  \mathrm{L}^\infty([0,T] \times \RR^d)$
and ({\bf Comp}) holds. The function $\hat{u}(t,x)$ is the unique
viscosity solution of (\ref{PDEpoly}). \label{thm}
\end{thm}

\begin{figure}
\begin{center}
\includegraphics[width=8cm,height=8cm]{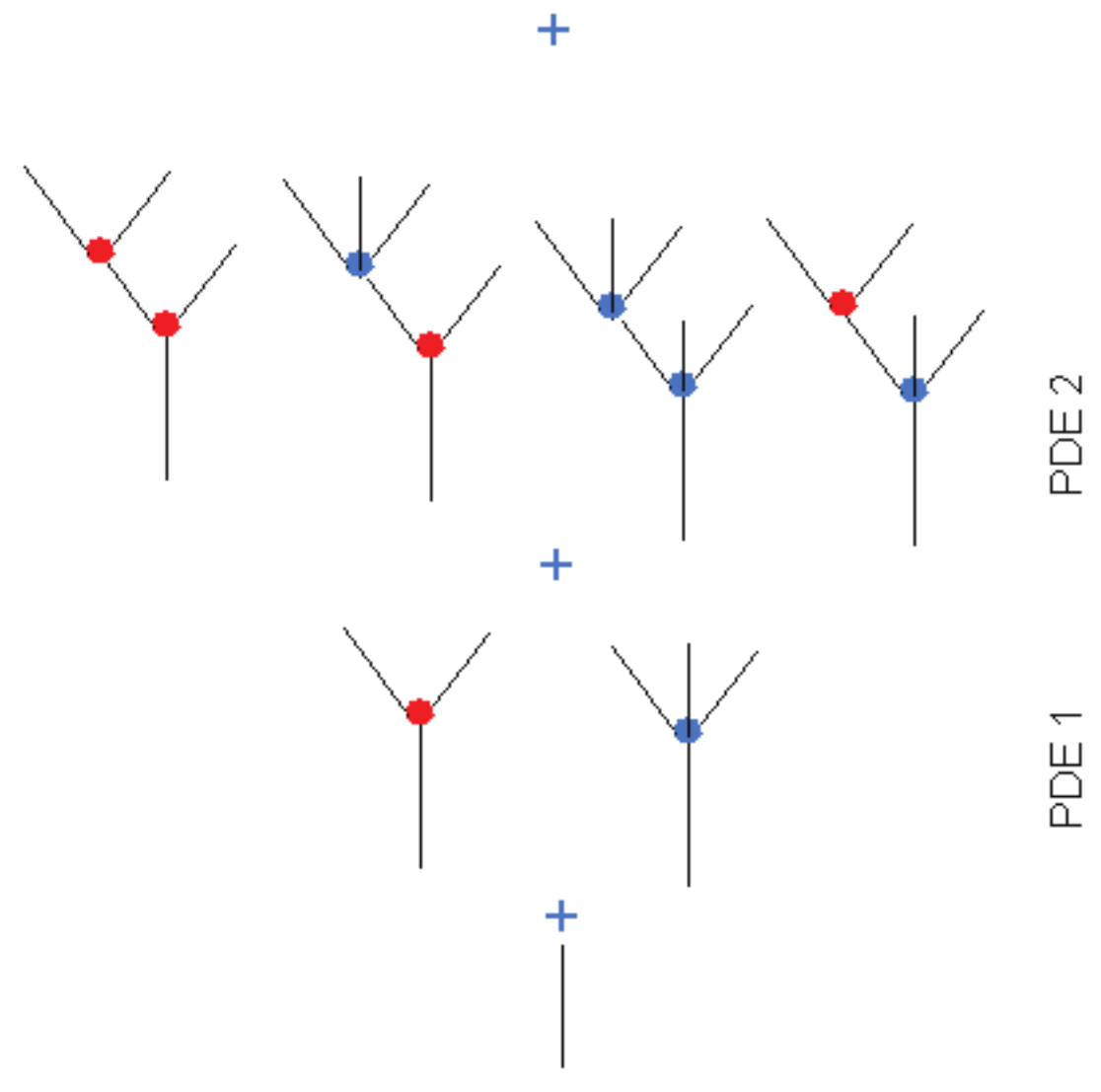}
\end{center}
  \caption{Marked Galton-Watson random tree for the non-linearity $F(u)={a \over p_2} p_2 u^2 + {b \over p_3} p_3 u^3$. The red (resp. blue) vertex corresponds to the weight
  ${a \over p_1}$ (resp. ${b \over p_2}$). The diagram with two red vertices has the weights $(\omega_1=2, \omega_2=0)$. } \label{Fey}
\end{figure}

\subsection*{Diagrammar interpretation} \no From Feynman-Kac's
formula, we have \bea u(t,x)=\EE_{t,x}[ 1_{\tau \geq T}\psi(X_T)]+
\EE_{t,x}[ F(u(\tau,X_\tau)) 1_{\tau<T}] \label{recursive} \eea \no
This integral equation can be recursively solved in terms of
multiple exponential random times $\tau_i$: \bea u(t,x)&=&\EE_{t,x}[
1_{\tau_0 \geq T}\psi(X_T)] \nonumber
\\&+& \EE_{t,x}[ F\left(\EE_{\tau_0}[ 1_{\tau_1\geq  T}\psi(X_T)]+
\EE_{\tau_0}[ F(\EE_{\tau_1}[ 1_{\tau_2\geq
T}\psi(X_T)])1_{\tau_1<T}]\right)1_{\tau_0<T}] + \cdots
\label{series} \eea Each term  can be interpreted as a Feynman
diagram (see Fig. \ref{Fey}) representing the trajectory of a
branching diffusion with a weight depending on the branching of each
monomial. For example in Fig. \ref{Fey}, the diagram with two red
vertices corresponds to \beaa \left({a_2 \over p_2}\right)^2
\EE_{t,x}[ 1_{\tau_0<T}\EE_{\tau_0}[1_{\tau_1 \geq T} \psi(X_T)]
\EE_{\tau_0}[1_{\tau_2<T}\EE_{\tau_2}[1_{\tau_3 \geq T}
\psi(X_T)]^2]]\eeaa \no By assuming that the series (\ref{series})
is convergent, one can guess that the solution is given by our
multiplicative functional (\ref{Rep1}).

\no In the next section, we focus on convergence issues and deduce a
sufficient condition to ensure that $\hat{u} \in
\mathrm{L}^\infty([0,T] \times \RR^d)$ if $\psi$ is bounded.

\subsection{Convergence issues}
\no The number of particles $N(\omega)$, produced by the branching
$\omega \equiv (\omega_0, \ldots, \omega_M)$, is \bea
N(\omega)=\sum_{k=0}^M (k-1) \omega_k +1 \eea The probability of
such a configuration satisfies the recurrence equation \bea
\PP(T|\omega)= \beta \sum_{k=0}^M&& \int_0^T dt \PP(t|\omega_0,
\ldots, \omega_k -1,
\ldots, \omega_M)N(\omega_0, \ldots, \omega_k -1, \ldots, \omega_M) \nonumber \\
&&p_k e^{-\beta k(T-t)} e^{-\beta (T-t)(N(\omega_0, \ldots, \omega_k
-1, \ldots, \omega_M)-1)} \label{rec}  \eea \no Indeed, if we have a
tree with a branching $(\omega_0, \ldots, \omega_k -1, \ldots,
\omega_M)$ at time $t$, a particle among the $N(\omega_0, \ldots,
\omega_k -1, \ldots, \omega_M)$ particles must die and produce $k$
descendants (with probability $p_k \beta e^{-k \beta (T-t)}$). The
remaining $N(\omega_0, \ldots, \omega_k -1, \ldots, \omega_M)-1$
particles must survive until maturity $T$ (with probability $
e^{-\beta (T-t)(N(\omega_0, \ldots, \omega_k -1, \ldots,
\omega_M)-1)} $).

\no We prove in the appendix that the Laplace transform of $\PP$,
$\hat{\PP}(T,c)=\EE[\prod_{k=0}^M e^{-c_k \omega_k}]$, satisfies the
equation  \bea &&\int_1^{\hat{\PP}(T,c)} {ds \over -s+ \sum_{k=0}^M
p_k e^{-c_k} s^{k} }=\beta T\;\;\mathrm{if}\;
\sum_{k=0}^M p_k e^{-c_k} \neq 1  \label{Laplace}\\
&&\hat{\PP}(T,c)=1\;\;\mathrm{if}\;\sum_{k=0}^M p_k e^{-c_k} =1 \eea
\no In the particular case of one branching type $k \neq 1$, we have
\beaa \hat{\PP}(T,c_k)={e^{c_k \over k-1} \over \left( 1-e^{\beta
T(k-1)}+e^{c_k+\beta T(k-1)} \right)^{1\over k-1}} \eeaa By assuming
that $\psi \in \mathrm{L}^\infty(\RR^d)$, the expectation in
(\ref{Rep1}) can then be bounded by \bea |\hat{u}(0,x)|& \leq
&\EE_{0,x}[ \prod_{k=0}^M \left({|a_k| \over p_k}\right)^{\omega_k}
||\psi||_\infty^{N(\omega)} ] = ||\psi||_\infty \hat
{\PP}\left(T,-\ln {|a_k|  \over p_k}-\ln ||\psi||_\infty^{k-1}
\right) \label{boundedu} \eea \no \no  from which we deduce  a
sufficient condition for convergence:
\begin{prop} Let us
assume that $\psi \in \mathrm{L}^\infty(\RR^d)$. Set
$p(s)=\beta\left( -s+ \sum_{k=0}^M |a_k| ||\psi||_\infty^{k-1}
s^{k}\right)$.
\begin{enumerate}
\item Case $\sum_{k=0}^M |a_k| ||\psi||_\infty^{k-1}>1$: We have $\hat{u} \in \mathrm{L}^\infty([0,T] \times \RR^d)$
(as defined by (\ref{Rep1})) if there exists $X \in \RR_+^*$ such
that \beaa \int_{1}^X {ds \over p(s)}= T \eeaa \no In the particular
case of one branching type $k$, the sufficient condition for
convergence reads as  \beaa |a_k| ||\psi||_\infty^{k-1} \left(
1-e^{-\beta T(k-1)} \right) < 1\eeaa
\item Case $\sum_{k=0}^M |a_k| ||\psi||_\infty^{k-1} \leq 1$: $\hat{u} \in  \mathrm{L}^\infty([0,T] \times \RR^d)$ for all $T$.
\end{enumerate}
\label{prop}
\end{prop} \no Note that our blow-up criteria does not depend on the
probabilities $p_k$ as expected.

\subsection{PDE (\ref{toy2})}
We assume that the function $(1-R)x^++Rx$ can be well approximated
by a polynomial $F(x)$ (see section \ref{cva}) and we consider the
PDE \beaa &&\partial_t u(t,x) + {\cal L} u(t,x) +{\beta  \over
1-R}\left( F(\EE_{t,x}[\psi(X_T)])-u(t,x) \right) =0,\quad
u(T,x)=\psi(x) \eeaa \no From Feynman-Kac's formula, we have \beaa
u(t,x)=\EE_{t,x}[ 1_{\tau \geq T}\psi(X_T)]+\EE_{t,x}[
F(\EE_\tau[\psi(X_T)] ) 1_{\tau<T}] \eeaa with $\tau$ a Poisson
default time with intensity $\beta/(1-R)$. As compared to the
previous section, we have the
 term $\EE_{t,x}[
F(\EE_\tau[\psi(X_T)] ) 1_{\tau<T}]$ instead of $\EE_{t,x}[
F(u(\tau,X_\tau)) 1_{\tau<T}]$. This term can be computed using the
previous algorithm by imposing that the particle can default only
once. This corresponds to the first three diagrams in Fig.
(\ref{Fey}). As $N_T$ is valued in $[0,M]$, our formula (\ref{Rep1})
is convergent here for all polynomial non-linearities.

\no As a conclusion, without any modification, the branching
particle algorithm can solve the two PDEs (\ref{toy})-(\ref{toy2})
modulo that the non-linearly $u^+$ can be fairly well approximated
by a polynomial.

\subsection{Optimal probabilities $p_k$} \label{sectionoptimal}
\no Is there a better choice than an uniform distribution $p_k={1
\over M+1}$ for improving the convergence?

\no For the PDE (\ref{toy}), the variance of the algorithm
(depending on the probabilities $p_k$) is bounded  by (see Equation
(\ref{boundedu})) \beaa ||\psi||_\infty \hat {\PP}\left(T,-2\ln
{|a_k| \over p_k}-2 \ln ||\psi||_\infty^{k-1} \right) \eeaa \no By
minimizing with respect to $p_k$, we get \bea p_k = {|a_k|
||\psi||_\infty^{k} \over \sum_{i=0}^M |a_i| ||\psi||_\infty^{i}}
\label{optimalp1} \eea \no Similarly, for the PDE (\ref{toy2}), the
variance (depending on the probabilities $p_k$) is bounded  by \beaa
\sum_{k=0}^M {a_k^2 \over p_k} ||\psi||_\infty^{2k} \beta T
e^{-\beta T} \eeaa \no By minimizing with respect to $p_k$, we get
also (\ref{optimalp1}).

\no We recall that the population in the Galton-Watson tree
disappears in finite time almost surely if $m \equiv \sum_{k=0}^M k
p_k \leq 1$ (see \cite{mel}). In the super-critical case $m>1$, the
population explodes at a finite time $T_\mathrm{exp}$ with
probability $1-s_0$ where $s_0=\inf \{ s \in [0,1], \sum_{k=0}^M p_k
s^k = s \}$. From (\ref{optimalp1}), we are in the super-critical
case if $\sum_{k=0}^M (k-1) |a_k| ||\psi||_\infty^k >0$.

\subsection{Numerical Experiments}

\no Before applying our algorithm to the problem of credit valuation
adjustment, we check it on polynomials which  do not belong to the
classes defined by (\ref{Poly}) and (\ref{Poly1}).

\subsubsection{Experiment 1}

\no We have implemented our algorithm for the two PDE types \beaa
\partial_t u+{\cal L} u + \beta(F(u) -u)=0,\quad u(T,x)=1_{x>1}  \;: \; \mathrm{PDE}2\eeaa and
\beaa
\partial_t u+{\cal L} u + \beta(F(\EE_{t,x}[1_{X_T>1}]) -u)=0,\quad u(T,x)=1_{x>1}  \;: \; \mathrm{PDE}1 \eeaa \no
with $F(u)={1 \over 2}\left(u^3-u^2\right)$.   ${\cal L}$ is the
It\^o generator of a geometric Brownian motion with a volatility
$\sigma_\mathrm{BS}=0.2$ and the Poisson intensity is $\beta=0.05$.
In financial terms, this corresponds to a CDS spread around $500$
basis points. The maturity is $T=10$ years. From (\ref{optimalp1}),
we note that our optimal probability distributions for PDE1 and PDE2
coincide with the uniform distribution. Moreover Proposition
(\ref{prop}) gives that the solution does not blow up.

\no The numerical method has been checked against a one-dimensional
PDE solver with a fully implicit scheme (see Table. \ref{BMMTest1})
for which we find $u=21.82\%$ (PDE1) and  $u=21.50\%$ (PDE2). Note
that this algorithm converges as expected and the error is properly
indicated by the Monte-Carlo standard deviation estimator (see
column Stdev).

\begin{table}\begin{center}
\begin{tabular}{|c|c|c|c|c|}
\hline N & Fair(PDE2) & Stdev(PDE2) & Fair(PDE1) & Stdev(PDE1)
\\ \hline
$12$ & $20.78$ & $0.78$ & $21.31$ & $0.79$ \\ \hline

$14$ & $22.25$ & $0.39$ & $21.37$ & $0.39$ \\ \hline

$16$ & $21.97$ & $0.19$ & $21.76$ & $0.20$ \\ \hline

$18$ & $21.90$ & $0.10$ & $21.51$ & $0.10$ \\ \hline

$20$ & $21.86$ & $0.05$ & $21.48$ & $0.05$ \\
\hline $22$ & ${\bf 21.81}$ & $0.02$ & ${\bf 21.50}$ & $0.02$
\\ \hline
\end{tabular}
\end{center}
\caption{MC price quoted in percent as a function of the number of
MC paths $2^N$. PDE pricer(PDE1) = ${\bf 21.82}$. PDE pricer(PDE2) =
${\bf 21.50}$. Non-linearity $F(u)={1 \over 2}\left(u^3-u^2\right)$.
} \label{BMMTest1}
\end{table}

\subsubsection{Experiment 2} \no Same test with $F(u)={1 \over 3}
\left(u^3-u^2-u^4\right)$ (see Table. \ref{BMMTest2}) and same
comments as above.

\begin{table}
\begin{center}
\begin{tabular}{|c|c|c|c|c|}
\hline N & Fair(PDE2) & Stdev(PDE2) & Fair(PDE1) & Stdev(PDE1)
\\ \hline
$12$ & $21.14$ & $0.78$& $20.00$ & $0.78$ \\ \hline

$14$ & $21.56$ & $0.38$ & $19.90$ & $0.39$ \\ \hline

$16$ & $21.62$ & $0.19$ & $20.25$ & $0.20$ \\ \hline

$18$ & $21.31$ & $0.10$ & $20.39$ & $0.10$ \\ \hline

$20$ & $21.38$ & $0.05$ & $20.36$ & $0.05$ \\ \hline

$22$ & ${\bf 21.36}$ & $0.02$ & ${\bf 20.40}$ & $0.02$ \\ \hline
\end{tabular}
\end{center}
\caption{MC price quoted in percent as a function of the number of
MC paths $2^N$.  PDE pricer(PDE1) = ${\bf 21.37}$. PDE pricer(PDE2)
= ${\bf 20.39}$. Non-linearity $F(u)={1 \over 3}
\left(u^3-u^2-u^4\right)$. } \label{BMMTest2}
\end{table}

\subsubsection{Experiment 3: Blow-up explosion}
\no It is well-known that the semi-linear PDE in $\RR^d$ \beaa
\partial_t u +{\cal L}u+u^2=0  \eeaa blows up in finite time
if $d\leq 2$ for any bounded positive payoff (see \cite{sug}). We
deduce that the PDE with the non-linearity $F(u)=u^2+u$ blows up in
finite time ($T_{\max}$) in one dimension. Using Proposition
(\ref{prop}), our sufficient condition reads as \beaa T_{\max}
||\psi||_\infty < 1 \eeaa \no We have verified this explosion when
the maturity $T$ is greater than $1$ year (in our case
$\psi=1_{x>0}$, $||\psi||_\infty=1$) using  our algorithm (and a PDE
solver as a benchmark). Note that for $T=1$, the algorithm  starts
to blow up (see Stdev = $0.49$). A different stochastic
representation can be obtained by setting $u=e^{(T-t)}v$. We get
\beaa
\partial_t v +{\cal L}v+e^{(T-t)}v^2-v=0  \;,\; v(T,x)=\psi(x)\eeaa
and this can be interpreted as a binary tree with a weight
$e^{(T-\tau)}$. Our stochastic representation gives then \bea
u(t,x)=e^{T-t} \EE_{t,x}\Big[ \prod_{i=1}^{N_T} \psi(z_T^i)
e^{\sum_{i=1}^{\sharp\mathrm{branching}}(T-\tau_i)} \Big]
\label{length} \eea where $\tau_i$ is the time where the $i$-th
branching appears. This representation (\ref{length}) appears in
\cite{mim} and was used to reproduce Sugitani's blow-up criteria
\cite{sug}.

\begin{table}
\begin{center}
\begin{tabular}{|c|c|c|}
\hline Maturity(Year) &  BBM alg.(Stdev) & PDE
\\ \hline
$0.5$ &  $71.66(0.09)$ & $71.50$ \\ \hline

$1$ &  $157.35(0.49)$ & $157.17$ \\ \hline

$1.1$ &  $\infty(\infty)$ & $\infty$ \\ \hline

\end{tabular}
\end{center}
\caption{MC price quoted in percent as a function of the maturity
for the non-linearity $F(u)=u^2+u$. $\psi(x) \equiv 1_{x>1}$.}
\label{Blowup}
\end{table}

\section{Credit valuation adjustment algorithm} \label{cva}
\no In the previous section, we have assumed that the payoff was
bounded: $\psi \in \mathrm{L}^\infty$. Then, the solution $u$ can
then be written as $v={u \over ||\psi||_\infty}$ where $v$ satisfies
\bea
\partial_t v +{\cal L}v+ \beta \left( v^+- v\right)=0,\quad ||v(T,\cdot)|| \leq 1
 \label{PDEvscaled}  \eea Therefore, by re-scaling, we can consider that the payoff
satisfies the condition $||\psi||_\infty \leq 1$. The condition
$\psi \in \mathrm{L}^\infty$ can be easily relaxed as observed in
(\cite{fah}, see Remark 3.7). Let $\psi$ be a payoff with
$\alpha$-exponential growth for some $\alpha>0$. We scale the
solution by an arbitrary smooth positive function $\rho$ given by
\beaa \rho(x)&\equiv& e^{\alpha
|x|} \; \mathrm{for} \; |x| \geq M \\
\tilde{v}(t,x) &\equiv& \rho^{-1}(x) v(t,x)\eeaa If we write the
linear operator $\cal L$ as ${\cal L}v=\mu(t,x) \partial_x v + {1
\over 2} \sigma^2(t,x) \partial_x^2 v$, then $\tilde{v}$ satisfies a
PDE\footnote{$\tilde{\cal L}$ is written in $d=1$. A similar
expression can be obtained in a multi-dimensional setup.} with the
same non-linearity $ \beta {v}^+$: \beaa
\partial_t \tilde{v} +\tilde{\cal L}\tilde{v}+ \beta \left(\tilde{v}^+-\tilde{v}\right)=0
\eeaa with $\tilde{\cal L} \tilde{v} =\left( \mu +\sigma^2 \rho^{-1}
\partial_x \rho \right) \partial_x \tilde{v}+ {1
\over 2} \sigma^2(t,x) \partial_x^2 \tilde{v}+\left( \mu \rho^{-1}
\partial_x \rho +{1 \over 2} \rho^{-1} \sigma^2
\partial_x^2 \rho \right)\tilde{v}$.

\no What remains to be done in order to use (\ref{Rep1}) is to
approximate $v^+$ by a polynomial $F(v)$: \bea
\partial_t v +{\cal L}v +\beta \left( F(v)-v \right)=0,\quad v(T,x)=\psi(x) \label{PDEv}
\eea \no In our numerical experiments, we take (see Fig.
\ref{Approxu}) \bea F(u)=0.0589 +0.5 u +0.8164 u^2 -0.4043 u^4
\label{choiceu} \eea Proposition \ref{prop} gives that the solution
does not blow up  if $\beta T<0.50829$ (Take $X=\infty$ with
$||\psi||_\infty=1$). Moreover, as a numerical check of
(\ref{boundedu}), we have computed using a PDE solver the solution
of (\ref{PDEv}) with $\psi(x)=1$, $\tilde{F}(u)=0.0589 +0.5 u
+0.8164 u^2 +0.4043 u^4$, $\beta=0.05$ and $T=10$ years. The
solution $X=\hat{\PP}\left(T,-\ln {|a_k| \over p_k} \right)$
coincides with our upper bound in (\ref{boundedu}) and should
satisfy \bea \int_{1}^X {ds \over -s+0.0589 +0.5 u +0.8164 u^2
+0.4043 u^4}= 0.5 \label{identity} \eea We found $X=4.497$ (PDE
solver) and the reader can check that this value satisfies the above
identity (\ref{identity}) as expected.

\begin{figure}
\begin{center}
\includegraphics[width=8cm,height=8cm]{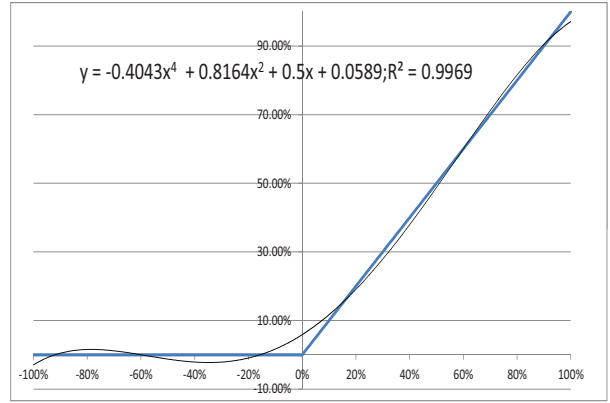}
\end{center}
  \caption{$u^+$ versus its polynomial approximation on $[-1,1]$.} \label{Approxu}
\end{figure}

\subsection{Algorithm: Final recipe}
The algorithm for solving PDEs (\ref{toy})-(\ref{toy2}) can be
described by the following steps:
\begin{enumerate}
\item Choose a polynomial approximation of $u^+ \simeq \sum_{k=0}^M
a_k u^k$ on the domain $[-1,1]$.
\item Simulate the assets and the Poisson default time with intensity $\beta$ (resp. ${\beta \over 1-R}$) for PDE2 (resp. PDE1).
Note that the intensity $\beta$ can be stochastic (Cox process),
usually calibrated to default probabilities implied from CDS market
quotes.
\item At each default time, produce $k$ descendants with probability
$p_k$ (given by (\ref{optimalp1})). For PDE type $2$, descendants,
produced after the first default, become immortal.
\item Evaluate for each particle alive the payoff \beaa &&  \prod_{i=1}^{N_T}
\psi(z_T^i) \prod_{k=0}^M \left({a_k \over p_k }\right)^{\omega_k}
 \;,\; {\bf \mathrm{PDE2}} \\
&&  \prod_{i}^{N_T \in [0,M]} \psi(z_T^i) \left({a_1(1-R)+R \over
p_1 }\right)^{\omega_1} \prod_{k \neq 1}^M \left({a_k (1-R)\over p_k
}\right)^{\omega_k}
 \;\;(\mathrm{here},\; \sum_{k=0}^M \omega_k =0 \; \mathrm{or} \; 1
 ) \;,\; {\bf \mathrm{PDE1}} \eeaa where $\omega_k$ denotes the number of branching type $k$. We should
 highlight that the algorithm for $\mathrm{PDE1}$ is always
 convergent for all $T$ whatever condition on the payoff as the multiplicative
 functional involves at most $M$ particles.
\end{enumerate}
\begin{Remark} \label{colla}
In the case of collateralized positions, the non-linearity $u_t^+$
should be substituted with $(u_t-u_{t+\Delta})^+$ where $\Delta$ is
a delay. Using our polynomial approximation, we get
$F(u_t-u_{t+\Delta})$. By expanding this function, we get monomials
of the form $\{u_t^p u_{t+\Delta}^q\}$. Our algorithm can then be
easily extended to handle this case. At each default time $\tau$, we
produce $p$ descendants starting at $(\tau,X_\tau)$ and $q$
descendants starting at $(\tau+\Delta,X_{\tau+\Delta})$.
\end{Remark}

\no A natural question is to characterize the error of the algorithm
as a function of the approximation error of $u^+$ by $F(u)$. Using
the parabolicity of the semi-linear PDE, we can characterize the
bias of our algorithm (the proof is reported in the appendix):
\begin{prop} Let us assume that $\underline{F}(v)$ and
$\overline{F}(v)$ are two polynomials satisfying ({\bf Comp}), the
sufficient condition in Prop. \ref{prop} for a maturity $T$ and
\beaa \underline{F}(x) \leq x^+ \leq \overline{F}(x) \eeaa \no We
denote $\underline{v}$ and $\overline{v}$ the corresponding
solutions of (\ref{PDEv}) and v the solution of
($\ref{PDEvscaled}$). Then \beaa \underline{v} \leq v \leq
\overline{v} \eeaa \label{propbias}
\end{prop}

\no A similar result can be found for PDE (\ref{toy2}). In the case
of American options, our algorithm gives robust lower and upper
bounds.

\subsection{Complexity}

\no By approximating $u^+$  with an infinite high-order polynomial -
say $N_2$ - our algorithm converges towards the brute force
``Monte-Carlo of Monte-Carlo"  method with a complexity $O(N_1
\times N_2)$. By comparison, with our choice (\ref{choiceu}), the
complexity is at most $O(4N_1)$ for PDE type (\ref{toy2}). Moreover,
this method allows to solve exactly PDE type (\ref{toy}), which can
not be tackled without relying on an approximation within  the
``Monte-Carlo of Monte-Carlo" method.

\subsection{Numerical examples}

\no We have implemented our algorithm for the two PDE types \beaa
\partial_t u+{1 \over 2 }x^2 \sigma^2_\mathrm{BS} \partial_x^2 u + \beta \left(u^+-u\right)=0,\quad u(T,x)=1-2.1_{x>1}  \;: \; \mathrm{PDE}1\eeaa and
\beaa
\partial_t u+{1 \over 2 }x^2 \sigma^2_\mathrm{BS} \partial_x^2 u + {\beta \over 1-R} \left(
(1-R)\EE_{t,x}[1-2.1_{X_T>1} ]^++R\EE_{t,x}[1-2.1_{X_T>1} ]-u\right)
=0,\quad \mathrm{PDE}2 \eeaa \no with Poisson intensities
$\beta=1\%$, $\beta=3\%$ and a recovery rate $R=0.4$ (see Tab.
\ref{BMMTest3}, \ref{BMMTest4}, \ref{BMMTest5}, \ref{BMMTest6}). In
financial term, this corresponds to CDS spreads around $100$ and
$300$ basis points. The method has been checked using a PDE solver
with the polynomial approximation (\ref{choiceu}) (see Column ``PDE
with poly."). In order to justify the validity of (\ref{choiceu}),
we have included the PDE price with the true non-linearity $u^+$
(see Column ``PDE"). As it can be observed, prices, produced by our
algorithm, converge to the PDE solver with the polynomial
approximation and are close to the exact CVA values. We would like
to highlight that replacing the Black-Scholes generator ${1 \over 2
}x^2 \sigma^2_\mathrm{BS}
\partial_x^2 $ by a multi-dimensional operator $\cal L$ can be easily handled
 in our framework by simulating the branching particles with a
diffusion process associated to $\cal L$. This is out-of-reach with
finite-difference scheme methods and not such an easy step for the
BSDE approach.

\begin{table}
\begin{center}
\begin{tabular}{|c|c|c|c|}
\hline Maturity(Year) & PDE with poly.   & BBM alg. & PDE
\\ \hline
$2$ &  $11.62$ & $11.63(0.00)$ & $11.62$ \\ \hline

$4$ & $16.54$  & $16.53(0.00)$ & $16.55$ \\ \hline

$6$ & $20.28$  & $20.27(0.00)$ & $20.30$ \\ \hline

$8$ &  $23.39$ & $23.38(0.00)$ & $23.41$ \\ \hline

$10$ &  $26.11$ & $26.09(0.00)$ & $26.14$ \\ \hline

\end{tabular}
\end{center}
\caption{MC price quoted in percent as a function of the maturity
for PDE 1 with $\beta=1\%$.} \label{BMMTest3}
\end{table}

\begin{table}
\begin{center}
\begin{tabular}{|c|c|c|c|}
\hline Maturity(Year) & PDE with poly.   & BBM alg.(Stdev) & PDE
\\ \hline

$2$ &  $11.62$ & $11.64(0.00)$ & $11.63$ \\ \hline

$4$ & $16.56$  & $16.55(0.02)$ & $16.57$ \\ \hline

$6$ & $20.32$  & $20.30(0.00)$ & $20.34$ \\ \hline

$8$ &  $23.45$ & $23.45(0.00)$ & $23.48$ \\ \hline

$10$ &  $26.20$ & $26.18(0.00)$ & $26.24$ \\ \hline
\end{tabular}
\end{center}
\caption{MC price quoted in percent as a function of the maturity
for PDE 2 with $\beta=1\%$.} \label{BMMTest4}
\end{table}

\begin{table}
\begin{center}
\begin{tabular}{|c|c|c|c|}
\hline Maturity(Year) & PDE with poly.  & BBM alg. & PDE
\\ \hline
$2$ & $12.34$  & $12.35(0.00)$ & $12.35$    \\ \hline

$4$ & $17.72$  & $17.71(0.00)$ & $17.75$    \\ \hline

$6$ & $21.77$  & $21.76(0.00)$ & $21.82$  \\ \hline

$8$ & $25.07$  & $25.06(0.00)$ & $25.14$ \\ \hline

$10$ & $27.89$  & $27.88(0.00)$ & $27.98$   \\ \hline
\end{tabular}
\end{center}
\caption{MC price quoted in percent as a function of the maturity
for PDE 1 with $\beta=3\%$.} \label{BMMTest5}
\end{table}

\begin{table}
\begin{center}
\begin{tabular}{|c|c|c|c|}
\hline Maturity(Year) & PDE with poly.   & BBM alg.(Stdev) & PDE
\\ \hline
$2$ & $12.38$  & $12.39(0.00)$ & $12.39$    \\ \hline

$4$ & $17.88$  & $17.86(0.00)$ & $17.91$    \\ \hline

$6$ & $22.08$  & $22.07(0.01)$ & $22.14$  \\ \hline

$8$ & $25.58$  & $25.57(0.01)$ & $25.66$ \\ \hline

$10$ & $28.62$  & $28.60(0.01)$ & $28.74$   \\ \hline
\end{tabular}
\end{center}
\caption{MC price quoted in percent as a function of the maturity
for PDE 2 with $\beta=3\%$.} \label{BMMTest6}
\end{table}

\section{Conclusion}
\no Credit valuation adjustment is now an important quantitative
issue which needs to receive special attention. The brute force
``Monte-Carlo of Monte-Carlo" or the BSDE approach is not, as it
looks like, a decent solution for multi-asset portfolios. We have
shown the efficiency of our algorithm based on  marked branching
diffusions on various numerical examples. This method can also be
used for semi-linear PDEs with polynomial non-linearities and
extended to fully non-linear PDEs by including in the branching
process Malliavin weights for derivatives. We left this
investigation for future research.

\no {\bf Acknowledgements.} The author wishes to thank the members
of the Global Markets Quantitative Research Group at Soci\'et\'e
G\'en\'erale  for their comments.  He is also grateful to
Jean-Fran\c{c}ois Delmas and Denis Talay  for useful discussions.

\section*{Appendix}

\begin{proof}[Proof of Theorem \ref{thm}]
The proof proceeds similarly as in subsection
\ref{branchingsection}. By using the independence and the strong
Markov property, we obtain \beaa \hat{u}(t,x)&=&\EE_{t,x}[1_{\tau
\geq T}\psi(z^1_T)] +\sum_{k=0}^M\EE_{t,x}[ 1_{\tau<T}  a_k
\prod_{j=1}^{k} \EE_\tau[ \prod_{i=1}^{N_T^j(\tau)} \prod_{k=0}^M
\left(a_k \over p_k\right)^{\omega_k^j}
\psi(z_T^{i,j,z_\tau})] \\
&=&\EE_{t,x}[1_{\tau \geq T}\psi(z^1_T)] +\EE_{t,x}[ 1_{\tau<T}
\sum_{k=0}^M
a_k \prod_{j=1}^{k} \hat{u}(\tau,z^1_\tau)] \\
&=&\EE_{t,x}[1_{\tau \geq T}\psi(z^1_T)] + \EE_{t,x}[
 F\left(\hat{u}(\tau,z^1_\tau)\right)1_{\tau<T}]  \\
 &=& \EE_{t,x}[ e^{-\int_t^T \beta(s)ds}
\psi(z_T^1)]+\int_t^T   \EE_{t,x}[\beta(s) e^{-\int_t^s \beta(u)du}
F\left(\hat{u}(s,z^1_s)\right)]  ds  \eeaa By assuming that $\hat{u}
\in \mathrm{L}^\infty([0,T] \times \RR^d)$, we deduce that $\hat{u}$
is a viscosity solution of  PDE (\ref{PDEpoly}) (see Theorem 6.4 in
\cite{tou}). The comparison result (Assumption ({\bf Comp})) implies
uniqueness, i.e. $u=\hat{u}$.
\end{proof}

\begin{proof}[Proof of formula \ref{Laplace}] We set
$\PP(T|\omega)={e^{-\beta T N(\omega)} } q(T|\omega)$ for
convenience. We get the relation \beaa q(T|\omega)= \beta
\sum_{k=0}^M&& \int_0^T dt q(t|\omega_0, \ldots, \omega_k -1,
\ldots, \omega_M) N(\omega_0, \ldots, \omega_k -1, \ldots, \omega_M)
p_k e^{\beta t (k-1)}   \eeaa which is equivalent to \beaa
\partial_T q(T|\omega)&=& \beta  \sum_{k=0}^M  q(T|\omega_0, \ldots,
\omega_k -1, \ldots, \omega_M) N(\omega_0, \ldots, \omega_k -1,
\ldots, \omega_M) p_k e^{\beta T (k-1)}  \\
q(0|\omega)&=&\delta_{\omega=0} \eeaa The Laplace transform of $q$,
$\hat{q}(T,c)\equiv \EE^q [\prod_{k=0}^M e^{-(k-1) c_k \omega_k }]$,
satisfies the first-order PDE \beaa
\partial_T \hat{q}(T|c)= \beta \sum_{k=0}^M p_k \left(  \hat{q}(T,c)-
\sum_{q=0}^M \partial_{c_q} \hat{q}(T,c) \right) e^{\left( \beta T
-c_k \right)(k-1)}  \;,\; \hat{q}(0|c)=1 \eeaa The solution is given
by \beaa \hat{q}(T|c)=e^{\left( c_0-c_0(T) \right)} \eeaa where the
coefficients $\{c_q(T)\}_{q=0,\ldots, M}$ are solutions of the ODEs
\beaa {dc_q(t) \over dt }=-\beta \sum_{k=0}^M p_k e^{\left( \beta
(T-t) -c_k(t) \right)(k-1)} \;,\; c_q(0)=c_q \eeaa \no The solution
is given by $c_q(t)=c_q -\beta t -\ln U(t|c)$ with \beaa {d U(t|c)
\over dt }=\beta \left( -U(t|c)+ \sum_{k=0}^M p_k e^{\left( \beta T
-c_k \right)(k-1)} U^{k}(t|c) \right) \;,\; U(0|c)=1\eeaa This gives
\beaa \hat{q}(T|c)&=&e^{\beta T} U(T|c)\;\;\mathrm{if}\;\sum_{k=0}^M
p_k e^{\left( \beta T -c_k
\right)(k-1)}\neq 1  \\
&=&e^{\beta T}\;\;\mathrm{if}\; \sum_{k=0}^M p_k e^{\left( \beta T
-c_k \right)(k-1)}= 1 \eeaa where $U(T|c)$ satisfies  \beaa
\int_1^{U(T|c)} {ds \over -s+ \sum_{k=0}^M p_k e^{\left( \beta T
-c_k \right)(k-1)} s^{k} }=\beta T\;\mathrm{if}\; \sum_{k=0}^M p_k
e^{\left( \beta T -c_k \right)(k-1)}\neq 1  \eeaa \no Finally, we
use that $\hat{\PP}(T|c)=U(T|{c_k \over k-1}+\beta T)$. \end{proof}

\begin{proof}[Proof of Proposition \ref{propbias}] The function $\delta=\bar{v}-v$ satisfies the linear PDE
\beaa
\partial_t \delta + {\cal L} \delta -\beta \delta + \beta \left(
\bar{v}^+ - {v}^+ \over \bar{v}-v\right)1_{v \neq \bar{v}}
\delta+\beta\left( {F}({\bar{v}})-\bar{v}^+\right)=0 \;,\;
\delta(T,x)=0 \eeaa Note that the term $r_t \equiv
1-\left({\bar{v}_t^+ -v_t^+ \over \bar{v}_t-v_t}\right)1_{v_t \neq
\bar{v}_t}$ is lower bounded. Feynman-Kac's formula gives \beaa
\delta(t,x)=\int_t^T \beta
\EE_{t,x}[\left({F}({\bar{v}})-\bar{v}^+\right)e^{-\beta \int_t^s
r_u du}] \eeaa from which we conclude the proof as $\bar{F}(x) \geq
x^+$ by assumption.
\end{proof}

\end{document}